\documentclass[onecolumn, a4size, 11pt]{IEEEtran}
\usepackage{amsmath}
\usepackage{amssymb}
\usepackage{amsfonts}
\usepackage{algorithm}
\usepackage{graphicx}
\usepackage{epsfig}
\usepackage{subfigure}
\usepackage{psfrag}
\usepackage{xcolor}

\title{Cross-Layer Design for Downlink Multi-Hop Cloud Radio Access Networks with Network Coding
\footnote{The authors are with The Edward S. Rogers Sr. Department
of Electrical and Computer Engineering, University of Toronto,
Toronto, Ontario, Canada, M5S3G4, (e-mail:lianguot.liu@utoronto.ca; weiyu@comm.utoronto.ca).}
\footnote{This work is supported by the Natural Sciences and Engineering Research Council (NSERC) of Canada through a Collaborative Research and Development (CRD) grant.}}

\author{Liang Liu, \IEEEmembership{Member,~IEEE} and Wei Yu, \IEEEmembership{Fellow,~IEEE}}

\begin{document}
\maketitle \thispagestyle{empty} \vspace{-0.3in}

\begin{abstract}
There are two fundamentally different fronthaul techniques in the downlink communication of cloud radio access network (C-RAN): the \emph{data-sharing strategy} and the \emph{compression-based strategy}. Under the former strategy, each user's message is multicast from the central processor (CP) to all the serving remote radio heads (RRHs) over the fronthaul network, which then cooperatively serve the users through joint beamforming; while under the latter strategy, the user messages are first beamformed then quantized at the CP, and the compressed signal is unicast to the corresponding RRH, which then decompresses its received signal for wireless transmission. Previous works show that in general the compression-based strategy outperforms the data-sharing strategy. This paper, on the other hand, points out that in a C-RAN model where the RRHs are connected to the CP via multi-hop routers, data-sharing can be superior to compression if the network coding technique is adopted for multicasting user messages to the cooperating RRHs, and the RRH's beamforming vectors, the user-RRH association, and the network coding design over the fronthaul network are jointly optimized based on the techniques of sparse optimization and successive convex approximation. This is in comparison to the compression-based strategy, where information is unicast over the fronthaul network by simple routing, and the RRH's compression noise covariance and beamforming vectors, as well as the routing strategy over the fronthaul network are jointly optimized based on the successive convex approximation technique. The observed gain in overall network throughput is due to that information multicast is more efficient than information unicast over the multi-hop fronthaul of a C-RAN.
\end{abstract}


\setlength{\baselineskip}{1.3\baselineskip}

\newtheorem{theorem}{Theorem}
\newtheorem{proposition}{Proposition}
\newtheorem{remark}{Remark}
\newcommand{\mv}[1]{\mbox{\boldmath{$ #1 $}}}

\section{Introduction}\label{sec:Introduction}
As a promising candidate for the 5G cellular roadmap, cloud radio access network (C-RAN) enables a centralized processing architecture, using multiple relay-like base stations (BSs), named remote radio heads (RRHs), to serve mobile users cooperatively under the coordination of a central processor (CP). In the downlink, the benefit of the C-RAN architecture arises from the ability to cooperatively transmit signals from RRHs to minimize the effect of interference. It is worth noting that messages intended for different users in the network originate from the CP. As a result, a key question is to decide the most effective way to convey the useful information about the user messages to the RRHs over the finite-capacity fronthaul links for wireless transmission so as to minimize the unwanted interference seen by the users.

In the literature, a considerable amount of effort has been dedicated to the efficient utilization of the fronthaul capacities in the downlink communication in C-RAN (see e.g., \cite{Simeone16} and the references therein). Among them, the data-sharing strategy and compression-based strategy have attracted a great deal of attention. Specifically, under the data-sharing strategy, the CP shares user messages with the RRHs over the fronthaul network, which then encode the user messages into wireless signals and cooperatively transmit them to users \cite{Gesbert11}--\cite{Zhang16}. Generally speaking, due to the finite-capacity fronthaul links, the message of each user can only be sent to a subset of RRHs for cooperative transmission. Consequently, the user-RRH association strategy plays an essential role on the downlink throughput achieved by the data-sharing strategy. In \cite{Yu14}, the reweighted $\ell_1$-norm based technique is employed to optimize the RRH's beamforming vectors and user-RRH association so as to balance between the cooperation gain over the wireless network as well as the data traffic over the fronthaul network.

Instead of sharing direct user messages, another approach for enabling cooperation is to centrally compute the beamformed signals to be transmitted by the RRHs at the CP. Under the compression-based strategy, the CP compresses these beamformed signals and sends the compressed \mbox{signals} to the corresponding RRHs over the fronthaul links for wireless transmission. However, the compression process at the CP introduces quantization noises that limit the system performance. In \cite{Simeone13}, the transmit covariance for the users and compression noise covariance for the RRHs are jointly optimized to maximize the weighted sum-rate of the users subject to the fronthaul capacity constraints.

\begin{figure}
\begin{center}
\scalebox{0.7}{\includegraphics*{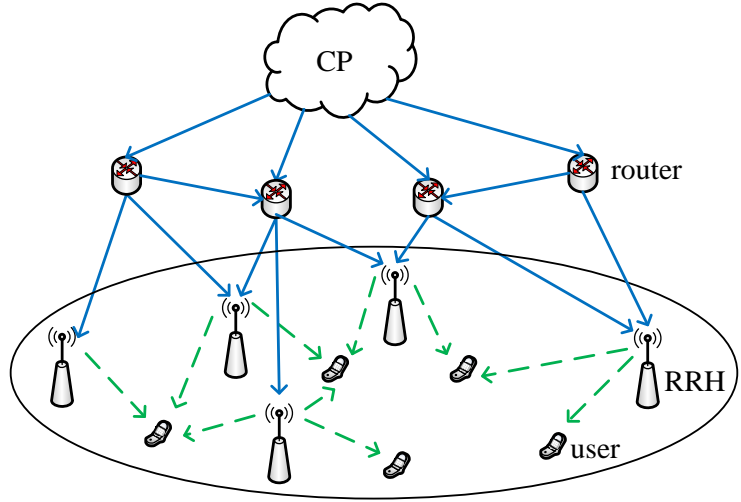}}
\end{center}
\caption{System model of downlink multi-hop C-RAN.}\label{fig1}
\end{figure}

Most previous works in this area focus on the beamforming and/or compression designs across the RRHs alone. However, besides the transmission strategy in the physical-layer, the routing strategy in the network-layer can significantly affect the throughput of downlink C-RAN as well, especially when the fronthaul network consists of edge routers and network processors over multiple hops, as illustrated in Fig. \ref{fig1}. This paper aims to jointly optimize the transmission and routing strategies in the downlink multi-hop C-RAN under both the data-sharing strategy and compression-based strategy and investigate which strategy achieves better throughput performance subject to the fronthaul capacity constraints. The main contributions of this paper are summarized as follows.
\begin{itemize}
\item This paper proposes a cross-layer framework to improve the throughput performance of the downlink multi-hop C-RAN, where the resources available in the physical-layer and network-layer are jointly optimized. Under the date-sharing strategy, a key observation is that such a cross-layer design provides an opportunity to leverage the network coding technique \cite{Yeung00} for multicasting user data to the corresponding RRHs over the multi-hop fronthaul network. A weighted sum-rate maximization problem is thus formulated, where RRH's beamforming vectors, user-RRH association, and network coding based routing are optimized in an overall design. Under the compression-based strategy, simple routing is used to unicast the compressed signal to each RRH. Weighted sum-rate maximization is formulated such that the RRH's compression noise covariance and beamforming vectors and the routing strategy are jointly optimized.
\item Efficient algorithms with monotonic convergence are proposed to solve the formulated weighted sum-rate maximization problems under the data-sharing strategy and compression-based strategy, respectively. \mbox{Specifically}, under the data-sharing strategy, we propose a two-stage algorithm to efficiently solve the studied problem by applying the techniques of sparse optimization and successive convex approximation: first, we approximate each user-RRH's discrete association indicator function by a continuous function and obtain a user-RRH association solution; then we fix this user-RRH association and find the corresponding beamforming and network coding strategy. Furthermore, under the compression-based strategy, a successive convex approximation based algorithm is proposed to solve the weighted sum-rate maximization problem. Both of the proposed algorithms are proved to yield locally optimal solutions that satisfy the Karush-Kuhn-Tucker (KKT) conditions of the studied problems.
\item By numerical results, it is shown that in the downlink multi-hop C-RAN, the data-sharing strategy can outperform the compression-based strategy in terms of throughput. This is because in the multi-hop fronthaul network, information multicast under the data-sharing strategy is more efficient than information unicast under the compression-based strategy. This complements the conclusions in \cite{Yu15,Yu16} which show that if the routing strategy is not considered, the compression-based strategy in general outperforms the data-sharing strategy in the downlink C-RAN in terms of the spectral efficiency and energy efficiency.
\end{itemize}

It is worth noting that under the data-sharing strategy, the joint beamforming and user-RRH association design in the downlink C-RAN has been previously investigated in \cite{Yu14}, but without considering the optimization of the routing strategy. Further, \cite{Luo14} proposes to jointly design the transmission and routing strategies in the downlink C-RAN, but in the model of \cite{Luo14} each user is solely served by one RRH, and the CP unicasts the data of each user to its associated RRH. Our paper differs from \cite{Yu14,Luo14} in allowing cooperative beamforming among RRHs and in the utilization of network coding technique over the fronthaul network for information multicast. Finally, under the compression-based strategy, the cross-layer design of the multi-hop C-RAN has been studied in the uplink in \cite{Simeone15}, where the RRHs utilize a compress-and-forward strategy. However, to the authors' best knowledge, the cross-layer design in the downlink multi-hop C-RAN has not been investigated prior to this work.

The rest of this paper is organized as follows. Section \ref{sec:System Model} presents the system model for the downlink multi-hop C-RAN. Sections \ref{sec:Data-Sharing Strategy} and \ref{sec:Compression-based Strategy} introduce the transmit and routing strategies under the data-sharing scheme and compression-based scheme, respectively. Section \ref{sec:Problem Formulation} formulates the weighted sum-rate maximization problems subject to the routing constraints for both schemes. Sections \ref{sec:Proposed Solution} and \ref{sec:Proposed Solution Compression} present the proposed solutions for the two formulated problems, respectively. Section \ref{sec:Numerical Results} provides numerical results to verify the effectiveness of the proposed cross-layer design and compares the performance between the data-sharing and compression-based strategies. Finally, Section \ref{sec:Conclusion} concludes the paper.

\section{System Model}\label{sec:System Model}
Consider the downlink communication in C-RAN where $N$ RRHs, denoted by the set $\mathcal{N}=\{1,\cdots,N\}$, cooperatively serve $K$ users, denoted by the set $\mathcal{K}=\{1,\cdots,K\}$, under the coordination of the CP. It is assumed that each RRH is equipped with $M\geq 1$ antennas, while each user is equipped with one single antenna. For the wireless network, it is assumed that the $N$ RRHs communicate with the $K$ users over quasi-static flat-fading channels over a given bandwidth of $B$ Hz. The channel from RRH $n$ to user $k$ is denoted by $\mv{h}_{k,n}\in \mathbb{C}^{M\times 1}$, $\forall n,k$. In this paper, it is assumed that the channels to all the $K$ users are perfectly known at the CP. Moreover, we assume that the CP and RRHs communicate over a multi-hop fronthaul network consisting of $J$ routers, denoted by the set $\mathcal{J}=\{1,\cdots,J\}$, and $L$ digital fronthaul links, denoted by the set $\mathcal{L}=\{1,\cdots,L\}$, as shown in Fig. \ref{fig1}. The capacity of each link $l\in \mathcal{L}$ is denoted by $C_l$ bits per second (bps).

This paper considers two fundamentally different fronthaul techniques, namely data-sharing strategy and compression-based strategy, in the downlink multi-hop C-RAN. Under the data-sharing strategy, the CP multicasts each user's message to all the serving RRHs via the multi-hop fronthaul network using the network coding technique \cite{Yeung00}, and each RRH then encodes the user messages into wireless signals and sends them to the users. Under the compression-based strategy, the CP first pre-forms and quantizes the beamformed signal for each RRH in an independent manner, then unicasts each RRH's compressed signal to the corresponding RRH by routing over the fronthaul network. Each RRH then decompresses its received signal and sends it to the users. In the following, we introduce in detail the proposed cross-layer architecture for the downlink multi-hop C-RAN under the data-sharing strategy and compression-based strategy, respectively.

\section{Data-Sharing Strategy}\label{sec:Data-Sharing Strategy}
In this section, we derive the throughput achieved by the data-sharing strategy in the downlink multi-hop C-RAN.

\subsection{Beamforming in the Physical-Layer}
With the data-sharing strategy, user messages are transmitted to the RRHs by the CP via the fronthaul network (refer to Section \ref{sec:Network coding over fronthaul network} for more detail). The equivalent baseband transmit signal of RRH $n$ is
\begin{align}\label{eqn:transmit signal scheme 1}
\mv{x}_n=\sum\limits_{k=1}^K\mv{w}_{k,n}s_k, ~~~ \forall n,
\end{align}where $s_k\sim \mathcal{CN}(0,1)$ denotes the message intended for user $k$, which is modeled as a circularly symmetric complex Gaussian (CSCG) random variable with zero-mean and unit-variance, and $\mv{w}_{k,n}\in \mathbb{C}^{M\times 1}$ denotes RRH $n$'s beamforming vector for user $k$. Suppose that RRH $n$ has a transmit sum-power constraint $P_n$; from (\ref{eqn:transmit signal scheme 1}), we have
\begin{align}\label{eqn:constraint 6}
\mathbb{E}[\mv{x}_n\mv{x}_n^H]=\sum_{k=1}^K\|\mv{w}_{k,n}\|^2\leq P_n, ~~~ \forall n.
\end{align}The received signal of user $k$ can be expressed as
\begin{align}\label{eqn:received signal scheme 1}
y_k=\sum\limits_{n=1}^N\mv{h}_{k,n}^H\mv{x}_n+z_k =\sum\limits_{n=1}^N\mv{h}_{k,n}^H\mv{w}_{k,n}s_k+\sum\limits_{n=1}^N\mv{h}_{k,n}^H\sum\limits_{i\neq k}\mv{w}_{i,n}s_i+z_k, ~ \forall k,
\end{align}where $z_k\sim \mathcal{CN}(0,\sigma^2)$ denotes the additive white Gaussian noise (AWGN) at user $k$.

The signal-to-interference-plus-noise ratio (SINR) for user $k$ is expressed as
\begin{align}\label{eqn:SINR scheme 1}
\gamma_k^{\rm DS}  = \frac{\left|\sum\limits_{n=1}^N\mv{h}_{k,n}^H\mv{w}_{k,n}\right|^2}{\sum\limits_{i\neq k}\left|\sum\limits_{n=1}^N\mv{h}_{k,n}^H\mv{w}_{i,n}\right|^2+\sigma^2}=  \frac{|\mv{h}_k^H\mv{w}_k|^2}{\sum\limits_{i\neq k}|\mv{h}_k^H\mv{w}_i|^2+\sigma^2}, ~ \forall k,
\end{align}where $\mv{h}_k=[\mv{h}_{k,1}^T,\cdots,\mv{h}_{k,N}^T]^T$ denotes the effective channel from all RRHs to user $k$, and $\mv{w}_k=[\mv{w}_{k,1}^T,\cdots,\mv{w}_{k,N}^T]^T$ denotes the effective beamforming vector for user $k$ across all RRHs. The achievable rate of user $k$ in bps under the data-sharing strategy is given by
\begin{align}\label{eqn:constraint 7}
r_k^{\rm DS}\leq B\log_2(1+\gamma_k^{\rm DS}), ~~~ \forall k.
\end{align}

\subsection{Network Coding in the Network-Layer}\label{sec:Network coding over fronthaul network}
Next, consider the data transmission from the CP to RRHs over the digital multi-hop fronthaul network. It is worth noting that if $\mv{w}_{k,n}\neq \mv{0}$, then user $k$ is served by RRH $n$; otherwise, user $k$ is not served by RRH $n$. As a result, we can define the user-RRH association indicator function $\alpha_{k,n}(\mv{w}_{k,n})$ as follows:
\begin{align}\label{eqn:user association}
\alpha_{k,n}(\mv{w}_{k,n})=\left\{\begin{array}{ll}1, & {\rm if} ~ \|\mv{w}_{k,n}\|^2\neq \mv{0}, \\ 0, & {\rm otherwise}, \end{array} \right. ~~~ \forall k,n.
\end{align}If user $k$ is served by RRH $n$, i.e., $\alpha_{k,n}(\mv{w}_{k,n})=1$, the CP needs to send the message $s_k$ to RRH $n$ over the multi-hop fronthaul network at a rate of $r_k^{\rm DS}$ bps; otherwise, the CP does not need to send $s_k$ to RRH $n$. To summarize, there are $K$ multicast sessions in the multi-hop fronthaul network, i.e., $s_1,\cdots,s_K$, and each session $s_k$ has a set $\mathcal{D}_k=\{n:\alpha_{k,n}(\mv{w}_{k,n})=1, n=1,\cdots,N\}$ of destinations.

The traditional approach for information multicast is to make each router replicate and forward its received information to the downstream routers. However, the optimization of such multicast routing is equivalent to the Steiner tree packing problem, which is NP-hard \cite{Li05,Li06}. Moreover, this replicate-and-forward based routing strategy is suboptimal since the coding operations at routers are necessary to achieve the multicast capacity \cite{Yeung00}. In this paper, we propose to apply the network coding technique to multicast each session to its destinations independently, but do not code between different sessions for the following reasons. First, this strategy results in an easy characterization of the routing region,  therefore making the optimal multicast routing problem polynomial time computable. Second, intersession coding provides marginal throughput gains over this approach \cite{Li05,Li06}.

Network coding allows flows for different destinations of a multicast session to share network capacity by being coded together. The pioneering work \cite{Yeung00} shows that for each single multicast session, the maximum multicast rate can be achieved for the entire multicast session if and only if it can be achieved for each multicast receiver independently. Moreover, with coding the actual physical flow on each link need only be the maximum of the individual destination's flows. As a result, the routing constraints for the multi-hop fronthaul network can be formulated as \cite{Li05,Li06}
\begin{align}
& \alpha_{k,n}(\mv{w}_{k,n})r_k^{\rm DS}\leq \sum\limits_{l\in \mathcal{I}(\mathcal{N}_n)}d^{k,n}_l, ~~~ \forall k, n, \label{eqn:constraint 1} \\
& \sum\limits_{l\in \mathcal{O}(\mathcal{J}_j)}d^{k,n}_l=\sum\limits_{l\in \mathcal{I}(\mathcal{J}_j)}d^{k,n}_l, ~~~ \forall k,n,j, \label{eqn:constraint 2} \\
& d^{k,n}_l\leq f^k_l, ~~~ \forall n,k,l, \label{eqn:constraint 3} \\
& \sum\limits_{k=1}^Kf^k_l\leq C_l, ~~~ \forall l, \label{eqn:constraint 4} \\
& f^k_l\geq 0, ~~~ d^{k,n}_l\geq 0, ~~~ \forall k,n,l, \label{eqn:constraint 5}
\end{align}where $d^{k,n}_l$ denotes the conceptual flow rate on link $l\in \mathcal{L}$ for the $k$th multicast session to its potential destination RRH $n$, $f^{k}_l$ denotes the actual flow rate on link $l$ for multicast session $k$, $\mathcal{N}_n$ and $\mathcal{J}_j$ denote RRH $n$ and router $j$, respectively, $\mathcal{I}(\mathcal{N}_n)$ denotes the set of links that are incoming to RRH $n$, and $\mathcal{I}(\mathcal{J}_j)$ and $\mathcal{O}(\mathcal{J}_j)$ denote the set of links that are incoming to and outgoing from router $j$, respectively. The first constraint guarantees that if $n\in \mathcal{D}_k$, then the $k$th session must flow at rate $r_k^{\rm DS}$ to its destination RRH $n$. The second constraint represents the law of flow conservation for conceptual flows. Note that the flow conservation constraint for the CP is not considered because it is automatically guaranteed by constraints (\ref{eqn:constraint 1}) and (\ref{eqn:constraint 2}). The third constraint indicates that the actual flow rate of the $k$th multicast session at each link $l$ is the maximum rate of the conceptual flows of that link to all the destinations, which is the benefit of network coding. The fourth constraint guarantees that the overall information flow rate at each link does not exceed the link capacity. The last constraint guarantees a positive flow rate for all the multicast sessions on all the links.\footnote{Given any flow rate solution satisfying constraints (\ref{eqn:constraint 1}) -- (\ref{eqn:constraint 5}), the code design which determines the content of each flow being transmitted across the network can be found according to \cite{Jaggi05,Wu03}.}

\section{Compression-based Strategy}\label{sec:Compression-based Strategy}
In this section, we derive the throughput achieved by the compression-based strategy in the downlink multi-hop C-RAN.

\subsection{Joint Beamforming and Quantization in the Physical-Layer}
Different from the above data-sharing strategy for which the user messages are sent to the RRHs for beamforming, under the compression-based strategy, the CP pre-forms the beamformed signal for each RRH instead. Similar to (\ref{eqn:transmit signal scheme 1}), the beamformed signal for RRH $n$ can be expressed as $\tilde{\mv{x}}_n=\sum_{k=1}^K\mv{w}_{k,n}s_k$, $\forall n$. Then, the CP compresses the beamformed signals and sends the quantization indices to the corresponding RRHs over the fronthaul network (please refer to Section \ref{sec:Routing over fronthaul network} for more information). The compression noise is modelled as a Gaussian random vector, i.e.,
\begin{align}\label{eqn:compressed signal}
\mv{x}_n=\tilde{\mv{x}}_n+\mv{e}_n=\sum\limits_{k=1}^K\mv{w}_{k,n}s_k+\mv{e}_n, ~~~ \forall n,
\end{align}where $\mv{e}_n\sim \mathcal{CN}(\mv{0},\mv{Q}_n)\in \mathbb{C}^{M\times 1}$, and $\mv{Q}_n\succeq \mv{0}$ denotes the covariance of the compression noise at RRH $n$.

Next, RRH $n$ transmits $\mv{x}_n$ to the users, $\forall n$. The transmit power constraint for RRH $n$ is then expressed as
\begin{align}\label{eqn:compression constraint 5}
\mathbb{E}[\mv{x}_n\mv{x}_n^H]=\sum_{k=1}^K\|\mv{w}_{k,n}\|^2+{\rm tr}(\mv{Q}_n)\leq P_n, ~~~ \forall n.
\end{align}

The baseband received signal at user $k$ is
\begin{align}\label{eqn:received signal compression}
y_k = \sum\limits_{n=1}^N\mv{h}_{k,n}^H\mv{x}_n+z_k= \sum\limits_{n=1}^N\mv{h}_{k,n}^H\mv{w}_{k,n}s_k +\sum\limits_{n=1}^N\mv{h}_{k,n}^H\sum\limits_{i\neq k}\mv{w}_{i,n}s_i +\sum\limits_{n=1}^N\mv{h}_{k,n}^H\mv{e}_n+z_k, ~ \forall k.
\end{align}The SINR of user $k$ is thus expressed as
\begin{align}\label{eqn:SINR compression}
\gamma_k^{\rm COM}  = \frac{\left|\sum\limits_{n=1}^N\mv{h}_{k,n}^H\mv{w}_{k,n}\right|^2}{\sum\limits_{i\neq k}\left|\sum\limits_{n=1}^N\mv{h}_{k,n}^H\mv{w}_{i,n}\right|^2+\sum\limits_{n=1}^N\mv{h}_{k,n}^H\mv{Q}_n\mv{h}_{k,n}+\sigma^2} =   \frac{|\mv{h}_k^H\mv{w}_k|^2}{\sum\limits_{i\neq k}|\mv{h}_k^H\mv{w}_i|^2+\sum\limits_{n=1}^N\mv{h}_{k,n}^H\mv{Q}_n\mv{h}_{k,n}+\sigma^2}, ~ \forall k.
\end{align}The achievable rate of user $k$ in bps under the compression-based strategy is given by
\begin{align}\label{eqn:compression constraint 6}
r_k^{\rm COM}\leq B\log_2(1+\gamma_k^{\rm COM}), ~~~ \forall k.
\end{align}

\subsection{Routing in the Network-Layer}\label{sec:Routing over fronthaul network}
In this paper we assume that the compression process is done independently across RRHs. According to the rate-distortion theory, the fronthaul capacity in bps required to convey the compressed signal $\mv{x}_n$ given in (\ref{eqn:compressed signal}) to RRH $n$ is expressed as
\begin{align}\label{eqn:compression rate}
T_n=BI(\mv{x}_n;\tilde{\mv{x}}_n) =B\log_2\left(\left|\sum\limits_{k=1}^K\mv{w}_{k,n}\mv{w}_{k,n}^H+\mv{Q}_n\right|/|\mv{Q}_n|\right), ~~~ \forall n.
\end{align}Note that instead of multicasting the information to the RRHs as in the data-sharing strategy, under the compression-based strategy, the CP merely unicasts each compressed signal $\mv{x}_n$ to its destination, i.e., RRH $n$. As a result, a simple routing strategy can be adopted for the information unicast over the fronthaul network. The routing constraints for the multihop fronthaul network $\mathcal{G}$ can then be formulated as
\begin{align}
& B\log_2\left(\left|\sum\limits_{k=1}^K\mv{w}_{k,n}\mv{w}_{k,n}^H+\mv{Q}_n\right|/|\mv{Q}_n|\right)\leq \sum\limits_{l\in \mathcal{I}(\mathcal{N}_n)}d^n_l, ~ \forall n, \label{eqn:compression constraint 1} \\
& \sum\limits_{l\in \mathcal{O}(\mathcal{J}_j)}d^n_l=\sum\limits_{l\in \mathcal{I}(\mathcal{J}_j)}d^n_l, ~~~ \forall n,j, \label{eqn:compression constraint 2} \\
& \sum\limits_{n=1}^Nd^n_l\leq C_l, ~~~ \forall l, \label{eqn:compression constraint 3} \\
& d^n_l\geq 0, ~~~ \forall n,l, \label{eqn:compression constraint 4}
\end{align}where $d^n_l$ denotes the flow rate on link $l\in \mathcal{L}$ for the $n$th unicast session , i.e., $\mv{x}_n$. The first constraint guarantees that the $n$th unicast session must flow at rate $T_n$ to its destination RRH $n$. The second constraint represents the law of flow conservation at each router. Note that the flow conservation constraint for the CP is not considered because it is automatically guaranteed by constraints (\ref{eqn:compression constraint 1}) and (\ref{eqn:compression constraint 2}).  The third constraint guarantees that the overall information flow rate at each link does not exceed the link capacity. The last constraint guarantees a positive flow rate for all the unicast sessions on all the links.
\begin{remark}\label{remark3}
By comparing Sections \ref{sec:Data-Sharing Strategy} and \ref{sec:Compression-based Strategy}, it can be observed that the key difference between the data-sharing strategy and compression-based strategy lies in how to utilize the fronthaul network. On one hand, user messages are transmitted over the fronthaul network with the former scheme, while compressed signals are transmitted with the latter scheme. On the other hand, the data-sharing strategy requires information multicast over the fronthaul network since each user's message is sent to all the RRHs serving this user, while the compression-based strategy merely requires information unicast since each RRH's compressed signal is sent to this RRH alone. Such different approaches generate different traffic in the fronthaul network, thus leading to different throughput in the considered multi-hop C-RAN, as will be shown in Section \ref{sec:Numerical Results}.
\end{remark}

\section{Problem Formulations}\label{sec:Problem Formulation}
In this paper, we aim to maximize the throughput of downlink multi-hop C-RAN via a joint optimization of the resources available in the physical-layer and network-layer under both the data-sharing strategy and the compression-based strategy.

\subsection{Data-Sharing Strategy}
For the data-sharing strategy introduced in Section \ref{sec:Data-Sharing Strategy}, we design the beamforming vectors at all RRHs, i.e., $\mv{w}_{k,n}$'s, and network coding strategy, i.e., $d^{k,n}_l$'s and $f^k_l$'s, to maximize the weighted sum-rate of all the users subject to each RRH's transmit power constraint over the wireless network as well as the network coding constraints in the multi-hop fronthaul network, i.e.,
\begin{subequations}\label{eqn:WSR maximization}\begin{align}\mathop{\mathrm{maximize}}_{\{\mv{w}_{k,n},r_k^{\rm DS},d^{k,n}_l,f^k_l\}} & ~ \sum\limits_{k=1}^K\mu_kr_k^{\rm DS} \label{eqn:WSR maximization sub}  \\
\mathrm {subject ~ to} ~~~~ & ~ (\ref{eqn:constraint 6}), ~ (\ref{eqn:constraint 7}), ~ (\ref{eqn:constraint 1})-(\ref{eqn:constraint 5}),
\end{align}\end{subequations}where $\mu_k>0$ denotes the positive rate weight for user $k$.

It is worth noting that without the routing constraints given in (\ref{eqn:constraint 1}) -- (\ref{eqn:constraint 5}), each user should be served by all the RRHs, i.e., $\alpha_{k,n}(\mv{w}_{k,n})=1$. However, with the constraints given in (\ref{eqn:constraint 1}) -- (\ref{eqn:constraint 5}), in general each RRH cannot support all the users in the downlink transmission, and as a result, from (\ref{eqn:user association}), for each RRH $n$, only a subset of users are associated with it, for which the corresponding user association function $\alpha_{k,n}(\mv{w}_{k,n})$ and beamforming vector $\mv{w}_{k,n}$ are non-zero. Moreover, the user association functions $\alpha_{k,n}(\mv{w}_{k,n})$'s also affect the network coding design since they determine the destinations of each multicast session. Therefore, the RRH's beamforming, user-RRH association, and network coding are coupled together and need to be jointly optimized in problem (\ref{eqn:WSR maximization}), which is a challenging problem in general.

It is also worth noting that constraint (\ref{eqn:constraint 1}) induces a \mbox{sparse} beamforming solution to problem (\ref{eqn:WSR maximization}). In the literature, \mbox{sparse} optimization technique has been previously used for the downlink beamforming design problem \cite{Yu16,Luo13}. Problem (\ref{eqn:WSR maximization}) differs from prior work in two aspects. First, \cite{Yu16,Luo13} encourage a sparse beamforming solution by penalizing the objective function with a sparsity term. However, problem (\ref{eqn:WSR maximization}) considered in this paper imposes a set of sparsity constraints which need to be strictly satisfied. Second, in \cite{Yu16,Luo13} the sparsity penalty is independent of the beamforming solution, but in constraint (\ref{eqn:constraint 1}) of our studied problem they are coupled. As a result, the existing sparse optimization techniques, e.g., least-absolute shrinkage and selection operator (LASSO), cannot be applied in this paper.

\subsection{Compression-based Strategy}
For the compression-based strategy introduced in Section \ref{sec:Compression-based Strategy}, we design the beamforming vectors at all RRHs, i.e., $\mv{w}_{k,n}$'s, compression noise covariance across the RRHs, i.e., $\mv{Q}_n$'s, and routing strategy, i.e., $d^n_l$'s, to maximize the weighted sum-rate of all the users subject to each RRH's transmit power constraint over the wireless network as well as the fronthaul capacity constraints in the multi-hop fronthaul network, i.e.,
\begin{subequations}\label{eqn:WSR maximization compression}\begin{align}\mathop{\mathrm{maximize}}_{\{\mv{w}_{k,n},r_k^{\rm COM},\mv{Q}_n,d^n_l\}} & ~ \sum\limits_{k=1}^K\mu_kr_k^{\rm COM} \label{eqn:WSR maximization sub compression}  \\
\mathrm {subject ~ to} ~~~~ & ~ (\ref{eqn:compression constraint 5}), ~ (\ref{eqn:compression constraint 6}), ~ (\ref{eqn:compression constraint 1})-(\ref{eqn:compression constraint 4}).
\end{align}\end{subequations}

It is worth noting that both the user rates given in (\ref{eqn:compression constraint 6}) and the fronthaul rates given in (\ref{eqn:compression rate}) are non-concave functions over the beamforming vectors $\mv{w}_{k,n}$'s and the compression noise covariance $\mv{Q}_n$'s. As a result, problem (\ref{eqn:WSR maximization compression}) is a non-convex optimization problem, and cannot be solved by the conventional convex optimization techniques.

In the following two sections, we propose efficient algorithms to obtain locally optimal solutions to the non-convex problems (\ref{eqn:WSR maximization}) and (\ref{eqn:WSR maximization compression}), respectively.

\section{Optimization of Data-Sharing Strategy}\label{sec:Proposed Solution}
In this section, we propose an efficient algorithm to solve problem (\ref{eqn:WSR maximization}) based on the techniques of sparse optimization as well as successive convex approximation. One main challenge for solving problem (\ref{eqn:WSR maximization}) is the discrete indicator function $\alpha_{k,n}(\mv{w}_{k,n})$ defined in (\ref{eqn:user association}). By applying standard sparse optimization technique, in this paper we use the following continuous function to approximate $\alpha_{k,n}(\mv{w}_{k,n})$:
\begin{align}\label{eqn:approximate function}
g_{\Phi}(\mv{w}_{k,n})=1-e^{-\Phi\|\mv{w}_{k,n}\|^2}, ~~~ \forall k,n,
\end{align}where $\Phi\gg 1$. It can be observed that when $\|\mv{w}_{k,n}\|^2=0$, then $g_{\Phi}(\mv{w}_{k,n})=\alpha_{k,n}(\mv{w}_{k,n})=0$. Otherwise, if $\|\mv{w}_{k,n}\|^2>0$, we have $g_{\Phi}(\mv{w}_{k,n})\rightarrow \alpha_{k,n}(\mv{w}_{k,n})=1$ with $\Phi \gg 1$.

%As a result, $g_{\Phi}(\mv{w}_{k,n})$ is a good approximation of the indicator function $\alpha_{k,n}(\mv{w}_{k,n})$, $\forall k,n$.

By using $g_{\Phi}(\mv{w}_{k,n})$ to approximate $\alpha_{k,n}(\mv{w}_{k,n})$, $\forall k,n$, problem (\ref{eqn:WSR maximization}) becomes the following continuous problem.
\begin{subequations}\label{eqn:WSR maximization 1}\begin{align}\hspace{-20pt}\mathop{\mathrm{maximize}}_{\{\mv{w}_{k,n},r_k^{\rm DS},d^{k,n}_l,f^k_l\}} & ~ \sum\limits_{k=1}^K\mu_kr_k^{\rm DS} \label{eqn:WSR maximization 1 sub}  \\
\mathrm {subject ~ to} ~~~~ & ~ g_{\Phi}(\mv{w}_{k,n})r_k^{\rm DS}\leq \sum\limits_{l\in \mathcal{I}(\mathcal{N}_n)}d^{k,n}_l, ~ \forall k, n, \label{eqn:constraint 11} \\ & ~ (\ref{eqn:constraint 6}), ~ (\ref{eqn:constraint 7}), ~ (\ref{eqn:constraint 2})-(\ref{eqn:constraint 5}).
\end{align}\end{subequations}

However, since $g_{\Phi}(\mv{w}_{k,n})$ is strictly less than one when $\|\mv{w}_{k,n}\|^2>0$, the solution to problem (\ref{eqn:WSR maximization 1}), which satisfies constraint (\ref{eqn:constraint 11}), may not satisfy constraint (\ref{eqn:constraint 1}) in problem (\ref{eqn:WSR maximization}). As a result, in this paper we propose to solve problem (\ref{eqn:WSR maximization}) in two steps as follows. First, we solve problem (\ref{eqn:WSR maximization 1}) and obtain the beamforming solution, denoted by $\hat{\mv{w}}_{k,n}$'s. The user-RRH association solution is then obtained as follows:
\begin{align}\label{eqn:opt user}
\alpha_{k,n}(\hat{\mv{w}}_{k,n})=\left\{\begin{array}{ll}1, & {\rm if} ~ g_{\Phi}(\hat{\mv{w}}_{k,n})\geq \psi, \\ 0, & {\rm otherwise},\end{array}\right. \forall k,n,
\end{align}where $0\leq \psi\leq 1$ is a threshold to control the user association solution.\footnote{In our simulation, we set $\Phi=50$ and $\psi=0.5$.} Second, we fix this user association solution in problem (\ref{eqn:WSR maximization}) and solve the following simplified problem to refine the beamforming and network coding strategy:
\begin{subequations}\label{eqn:WSR maximization fixed association}\begin{align}\hspace{-9pt}\mathop{\mathrm{maximize}}_{\{\mv{w}_{k,n},r_k^{\rm DS},d^{k,n}_l,f^k_l\}} & ~ \sum\limits_{k=1}^K\mu_kr_k^{\rm DS} \label{eqn:WSR maximization fixed association sub}  \\
\mathrm {subject ~ to} ~~~~ & ~ \alpha_{k,n}(\hat{\mv{w}}_{k,n})r_k^{\rm DS}\leq \sum\limits_{l\in \mathcal{I}(\mathcal{N}_n)}d^{k,n}_l, ~ \forall k,n, \label{eqn:constraint 1 fixed association} \\ & \|\mv{w}_{k,n}\|^2=0, ~ \forall  ~ \alpha_{k,n}(\hat{\mv{w}}_{k,n})=0,  \label{eqn:constraint 2 fixed association} \\ & ~ (\ref{eqn:constraint 6}), ~ (\ref{eqn:constraint 7}), ~ (\ref{eqn:constraint 2})-(\ref{eqn:constraint 5}).
\end{align}\end{subequations}In the following, we show how to solve problems (\ref{eqn:WSR maximization 1}) and (\ref{eqn:WSR maximization fixed association}), respectively.

\subsection{The First Stage: Solution to Problem (\ref{eqn:WSR maximization 1})}
Problem (\ref{eqn:WSR maximization 1}) is a non-convex problem due to constraints (\ref{eqn:constraint 7}) and (\ref{eqn:constraint 11}). As a result, the conventional convex optimization technique cannot be directly applied. In this section, we propose an efficient algorithm to solve problem (\ref{eqn:WSR maximization 1}) suboptimally based on the technique of successive convex approximation.

First, we consider constraint (\ref{eqn:constraint 7}), which is equivalent to
\begin{align}\label{eqn:constraint 7 eqv}
\frac{|\mv{h}_k^H\mv{w}_k|^2}{\sum\limits_{i\neq k}|\mv{h}_k^H\mv{w}_i|^2+\sigma^2}\geq 2^{\frac{r_k^{\rm DS}}{B}}-1, ~~~ \forall k.
\end{align}By introducing a set of auxiliary variables $\eta_k\geq 0$'s, $k=1,\cdots,K$, it can be shown that constraint (\ref{eqn:constraint 7 eqv}) is equivalent to the following two constraints:
\begin{align}
&  \mv{h}_k^H\mv{w}_k\geq \sqrt{(2^{\frac{r_k^{\rm DS}}{B}}-1)\eta_k}, ~~~ \forall k, \label{eqn:constraint 7 eqv 1} \\
&  \sqrt{\sum\limits_{i\neq k}|\mv{h}_k^H\mv{w}_i|^2+\sigma^2}\leq \sqrt{\eta_k}, ~~~ \forall k. \label{eqn:constraint 7 eqv 2}
\end{align}As a result, $\eta_k$ can be interpreted as the interference constraint for user $k$. Constraint (\ref{eqn:constraint 7 eqv 2}) can be further transformed into the following convex second-order cone (SOC) constraint:
\begin{align}\label{eqn:P1 constraint 2}
\left\|[\mv{h}_k^H\mv{w}_1,\cdots,\mv{h}_k^H\mv{w}_{k-1},\mv{h}_k^H\mv{w}_{k+1},\cdots,\mv{h}_k^H\mv{w}_K]^T\right\|
\leq \sqrt{\eta_k-\sigma^2}, ~~~ \forall k.
\end{align}For constraint (\ref{eqn:constraint 7 eqv 1}), $\sqrt{(2^{r_k^{\rm DS}/B}-1)\eta_k}$ is not a convex function. However, given any $\tilde{\beta}_k$, the following convex function is an upper bound for $\sqrt{(2^{r_k^{\rm DS}/B}-1)\eta_k}$:
\begin{align}\label{eqn:app}
\hspace{-5pt}f_{\tilde{\beta}_k}(r_k^{\rm DS},\eta_k)\hspace{-2pt}=\hspace{-2pt}\frac{\tilde{\beta}_k\eta_k}{2}+\frac{2^{\frac{r_k^{\rm DS}}{B}}-1}{2\tilde{\beta}_k}\geq \sqrt{(2^{\frac{r_k^{\rm DS}}{B}}-1)\eta_k}, ~ \forall k,
\end{align}where the equality holds if and only if $\tilde{\beta}_k=\sqrt{(2^{r_k^{\rm DS}/B}-1)/\eta_k}$. As a result, we use the following convex constraint to approximate constraint (\ref{eqn:constraint 7 eqv 1}):
\begin{align}\label{eqn:P1 constraint 1}
\mv{h}_k^H\mv{w}_k \geq \frac{\tilde{\beta}_k\eta_k}{2}+\frac{2^{\frac{r_k^{\rm DS}}{B}}-1}{2\tilde{\beta}_k},  ~ \forall k.
\end{align}

After approximating the non-convex constraint (\ref{eqn:constraint 7}) by the convex ones (\ref{eqn:P1 constraint 2}) and (\ref{eqn:P1 constraint 1}), we come to constraint (\ref{eqn:constraint 11}). First, we take the natural logarithm of the left-hand side (LHS) and right-hand side (RHS) of inequality constraint (\ref{eqn:constraint 11}), which results in
\begin{align}\label{eqn:constraint 1 log}
\log(1-e^{-\Phi\|\mv{w}_{k,n}\|^2})+\log(r_k^{\rm DS}) \leq \log\left(\sum\limits_{l\in \mathcal{I}(\mathcal{N}_n)}d^{k,n}_l\right), ~ \forall k, n.
\end{align}It can be shown that $\log(\sum_{l\in \mathcal{I}(\mathcal{N}_n)}d^{k,n}_l)$ is a concave function over $d^{k,n}_l$'s. However, the LHS of constraint (\ref{eqn:constraint 1 log}) is still non-convex. Since $\log(1-e^{-\Phi x})$ is a concave function over $x$, its first-order approximation serves as its upper bound. Specifically, given any $\tilde{x}$, the first-order approximation of $\log(1-e^{-\Phi x})$ can be expressed as
\begin{align}\label{eqn:first-order approximation}
\log(1-e^{-\Phi x}) \leq \frac{\Phi e^{-\Phi \tilde{x}}(x-\tilde{x})}{1-e^{-\Phi\tilde{x}}}+\log(1-e^{-\Phi \tilde{x}}),
\end{align}where the equality holds if and only if $x=\tilde{x}$. By substituting $x$ with $\|\mv{w}_{k,n}\|^2$, given any $\tilde{\mv{w}}_{k,n}$, a convex upper bound for $\log(1-e^{-\Phi\|\mv{w}_{k,n}\|^2})$ is expresses as
\begin{align}\label{eqn:constraint 1 log eqv}
\log(1-e^{-\Phi\|\mv{w}_{k,n}\|^2}) \leq \frac{\Phi e^{-\Phi \|\tilde{\mv{w}}_{k,n}\|^2}\|\mv{w}_{k,n}\|^2}{1-e^{-\Phi\|\tilde{\mv{w}}_{k,n}\|^2}}+\phi(\tilde{\mv{w}}_{k,n}), ~\forall k,n,
\end{align}where
\begin{align*}
\phi(\tilde{\mv{w}}_{k,n})=-\frac{\Phi e^{-\Phi \|\tilde{\mv{w}}_{k,n}\|^2}\|\tilde{\mv{w}}_{k,n}\|^2}{1-e^{-\Phi\|\tilde{\mv{w}}_{k,n}\|^2}}+\log(1-e^{-\Phi \|\tilde{\mv{w}}_{k,n}\|^2}).
\end{align*}The equality holds if and only if $\mv{w}_{k,n}=\tilde{\mv{w}}_{k,n}$.

Similarly, given any point $\tilde{r}_{k}^{\rm DS}$, the concave function $\log(r_k^{\rm DS})$ can be approximated by its first-order approximation as follows:
\begin{align}\label{eqn:first-order approximation 2}
\log(r_k^{\rm DS})\leq \frac{r_k^{\rm DS}-\tilde{r}_k^{\rm DS}}{\tilde{r}_k^{\rm DS}}+\log(\tilde{r}_k^{\rm DS}), ~~~ \forall k,
\end{align}where the equality holds if and only if $r_k^{\rm DS}=\tilde{r}_k^{\rm DS}$.

With (\ref{eqn:constraint 1 log eqv}) and (\ref{eqn:first-order approximation 2}), the non-convex constraint (\ref{eqn:constraint 1 log}) can be approximated by the following convex constraint:
\begin{align}\label{eqn:constraint 1 new 1}
\frac{\Phi e^{-\Phi \|\tilde{\mv{w}}_{k,n}\|^2}\|\mv{w}_{k,n}\|^2}{1-e^{-\Phi\|\tilde{\mv{w}}_{k,n}\|^2}}+ \frac{r_k^{\rm DS}-\tilde{r}_k^{\rm DS}}{\tilde{r}_k^{\rm DS}}+\phi(\tilde{\mv{w}}_{k,n})+\log(\tilde{r}_k^{\rm DS}) \leq \log\left(\sum\limits_{l\in \mathcal{I}(\mathcal{N}_n)}d^{k,n}_l\right), ~~~ \forall k, n.
\end{align}

To summarize, given $\tilde{r}_k^{\rm DS}$'s, $\tilde{\mv{w}}_{k,n}$'s, and $\tilde{\beta}_k$'s, the non-convex constraints (\ref{eqn:constraint 7}) and (\ref{eqn:constraint 11}) in problem (\ref{eqn:WSR maximization 1}) are approximated by the convex constraints given in (\ref{eqn:P1 constraint 2}), (\ref{eqn:P1 constraint 1}), and (\ref{eqn:constraint 1 new 1}). As a result, with any given $\tilde{r}_k^{\rm DS}$'s, $\tilde{\mv{w}}_{k,n}$'s, and $\tilde{\beta}_k$'s, problem (\ref{eqn:WSR maximization 1}) is approximated by the following convex problem.
\begin{subequations}\label{eqn:P1}\begin{align}\mathop{\mathrm{maximize}}_{\{\mv{w}_{k,n},r_k^{\rm DS},\eta_k,d^{k,n}_l,f^k_l\}} & ~ \sum\limits_{k=1}^K\mu_kr_k^{\rm DS} \label{eqn:P1 sub}  \\
\mathrm {subject ~ to} ~~~~~ & ~ (\ref{eqn:constraint 6}),  (\ref{eqn:P1 constraint 2}),  (\ref{eqn:P1 constraint 1}),  (\ref{eqn:constraint 1 new 1}), (\ref{eqn:constraint 2}) - (\ref{eqn:constraint 5}).
\end{align}\end{subequations}Since problem (\ref{eqn:P1}) is a convex problem, it can be globally solved by CVX \cite{Boyd11}. The successive convex approximation method based algorithm to problem (\ref{eqn:WSR maximization 1}) is summarized in Algorithm \ref{table1}, which iteratively updates $\tilde{r}_k^{\rm DS}$'s, $\tilde{\mv{w}}_{k,n}$'s, and $\tilde{\beta}_k$'s based on the solution to problem (\ref{eqn:P1}) as shown in Step 2). The convergence behaviour of Algorithm \ref{table1} is guaranteed in the following proposition.

\begin{algorithm}[t]
{\bf Initialization}: Set the initial values for $\tilde{\mv{w}}_{k,n}$'s, $\tilde{r}_k^{\rm DS}$'s, and $\tilde{\beta}_k$'s and set $t=1$; \\
{\bf Repeat}:
\begin{enumerate}
\item Find the optimal solution to problem (\ref{eqn:P1}) using CVX as $\{\mv{w}_{k,n}^{(t)},(r_k^{\rm DS})^{(t)},\eta_k^{(t)},(d^{k,n}_l)^{(t)},(f^k_l)^{(t)}\}$;
\item Update $\tilde{\mv{w}}_{k,n}=\mv{w}_{k,n}^{(t)}$, $\tilde{r}_k^{\rm DS}=(r_k^{\rm DS})^{(t)}$, and $\tilde{\beta}_k=\sqrt{(2^{(r_k^{\rm DS})^{(t)}/B}-1)/\eta_k^{(t)}}$, $\forall k,n$;
\item $t=t+1$.
\end{enumerate}
{\bf Until} convergence
\caption{Proposed Algorithm for Solving Problem (\ref{eqn:WSR maximization 1})}
\label{table1}
\end{algorithm}

\begin{proposition}\label{proposition1}
Monotonic convergence of Algorithm \ref{table1} is guaranteed, i.e., $\sum_{k=1}^K\mu_k(r_k^{\rm DS})^{(t)}\geq \sum_{k=1}^K \mu_k(r_k^{\rm DS})^{(t-1)}$. Moreover, the converged solution satisfies all the constraints as well as the KKT conditions of problem (\ref{eqn:WSR maximization 1}).
\end{proposition}

\begin{proof}
Please refer to Appendix \ref{appendix1}.
\end{proof}

\subsection{The Second Stage: Solution to Problem (\ref{eqn:WSR maximization fixed association})}
Given the user association in problem (\ref{eqn:WSR maximization fixed association}), constraint (\ref{eqn:constraint 1 fixed association}) becomes convex. By using (\ref{eqn:P1 constraint 2}) and (\ref{eqn:P1 constraint 1}) to approximate the non-convex constraint (\ref{eqn:constraint 7}), given any $\tilde{\beta}_k$'s, problem (\ref{eqn:WSR maximization fixed association}) can be approximated by the following convex problem.
\begin{subequations}\label{eqn:WSR maximization fixed association 1}\begin{align}\mathop{\mathrm{maximize}}_{\{\mv{w}_{k,n},r_k^{\rm DS},d^{k,n}_l,f^k_l\}} & ~ \sum\limits_{k=1}^K\mu_kr_k^{\rm DS} \label{eqn:WSR maximization fixed association 1 sub}  \\
\mathrm {subject ~ to} ~~~~ &  ~ \|\mv{w}_{k,n}\|^2\leq 0, ~ \forall \alpha_{k,n}(\hat{\mv{w}}_{k,n})=0, \label{eqn:constraint 2 fixed association 1} \\ & ~ (\ref{eqn:constraint 6}),  (\ref{eqn:P1 constraint 2}),  (\ref{eqn:P1 constraint 1}),  (\ref{eqn:constraint 1 fixed association}),  (\ref{eqn:constraint 2})-(\ref{eqn:constraint 5}).
\end{align}\end{subequations}Since problem (\ref{eqn:WSR maximization fixed association 1}) is a convex problem, it can be efficiently solved. The successive convex approximation based algorithm to problem (\ref{eqn:WSR maximization fixed association}) is summarized in Algorithm \ref{table2}. Similar to Proposition \ref{proposition1}, the convergence behaviour of Algorithm \ref{table2} is guaranteed in the following proposition.

\begin{algorithm}[t]
{\bf Initialization}: Set the initial values for $\tilde{\beta}_k$'s and set $t=1$; \\
{\bf Repeat}:
\begin{enumerate}
\item Find the optimal solution to problem (\ref{eqn:WSR maximization fixed association 1}) using CVX as $\{\mv{w}_{k,n}^{(t)},(r_k^{\rm DS})^{(t)},\eta_k^{(t)},(d^{k,n}_l)^{(t)},(f^k_l)^{(t)}\}$;
\item Update $\tilde{\beta}_k=\sqrt{(2^{(r_k^{\rm DS})^{(t)}/B}-1)/\eta_k^{(t)}}$, $\forall k,n$;
\item $t=t+1$.
\end{enumerate}
{\bf Until} convergence
\caption{Proposed Algorithm for Solving Problem (\ref{eqn:WSR maximization fixed association})}
\label{table2}
\end{algorithm}

\begin{proposition}\label{proposition2}
Monotonic convergence of Algorithm \ref{table2} is guaranteed, i.e., $\sum_{k=1}^K\mu_k(r_k^{\rm DS})^{(t)}\geq \sum_{k=1}^K \mu_k(r_k^{\rm DS})^{(t-1)}$. Moreover, the converged solution satisfies all the constraints as well as the KKT conditions of problem (\ref{eqn:WSR maximization fixed association}).
\end{proposition}

The overall two-stage algorithm to problem (\ref{eqn:WSR maximization}) is summarized in Algorithm \ref{table3}.
\begin{algorithm}[t]
\begin{enumerate}
\item Solve problem (\ref{eqn:WSR maximization 1}) based on Algorithm \ref{table1} and obtain the user-RRH association according to (\ref{eqn:opt user});
\item Solve problem (\ref{eqn:WSR maximization fixed association}) based on Algorithm \ref{table2} and obtain the beamforming and network coding solution.
\end{enumerate}
\caption{Overall Algorithm for Solving Problem (\ref{eqn:WSR maximization})}
\label{table3}
\end{algorithm}

\begin{remark}\label{remark2}
It is worth noting that \cite{Yu14} studies a similar problem of jointly optimizing the user-RRH association with the beamforming vectors. To deal with the discrete user-RRH association indicator functions (\ref{eqn:user association}), in \cite{Yu14} the reweighted $\ell_1$-norm technique is employed to approximate the fronthaul constraint (\ref{eqn:constraint 1}) by a set of weighted per-RRH power constraints. Then, an alternating optimization based iterative algorithm is proposed to find a beamforming and user-RRH association solution. Although the algorithm in \cite{Yu14} works well in practice, a rigorous convergence proof is not available. In contrast, the algorithm proposed in this paper always converge, but the performance depends on the tuning of the approximation parameters $\Phi$ and $\psi$.
\end{remark}

\section{Optimization of Compression-based Strategy}\label{sec:Proposed Solution Compression}
In this section, we propose an efficient algorithm to solve problem (\ref{eqn:WSR maximization compression}) based on the technique of successive convex approximation. There are two challenges to solve problem (\ref{eqn:WSR maximization compression}): the non-convex user rate constraint given in (\ref{eqn:compression constraint 6}) and fronthaul constraint given in (\ref{eqn:compression constraint 1}). In the following, we show how to circumvent the above two challenges.

First, similar to Section \ref{sec:Proposed Solution}, by introducing a set of auxiliary variables $\eta_k\geq 0$'s, $k=1,\cdots,K$, it can be shown that constraint (\ref{eqn:compression constraint 6}) is equivalent to the following two constraints:
\begin{align}
&  \mv{h}_k^H\mv{w}_k\geq \sqrt{(2^{\frac{r_k^{\rm COM}}{B}}-1)\eta_k}, ~~~ \forall k, \label{eqn:compression constraint 6 eqv 1} \\
&  \sqrt{\sum\limits_{i\neq k}|\mv{h}_k^H\mv{w}_i|^2+\sum\limits_{n=1}^N\mv{h}_{k,n}^H\mv{Q}_n\mv{h}_{k,n}+\sigma^2}\leq \sqrt{\eta_k}, ~~~ \forall k. \label{eqn:compression constraint 6 eqv 2}
\end{align}

Constraint (\ref{eqn:compression constraint 6 eqv 2}) can be further transformed into the following convex SOC constraint:
\begin{align}\label{eqn:P3 constraint 2}
\left\|[\mv{h}_k^H\mv{w}_1,\cdots,\mv{h}_k^H\mv{w}_{k-1},\mv{h}_k^H\mv{w}_{k+1},\cdots,\mv{h}_k^H\mv{w}_K]^T\right\|
\leq \sqrt{\eta_k-\sum\limits_{n=1}^N{\rm tr}(\mv{H}_{k,n}\mv{Q}_n)-\sigma^2}, ~~~ \forall k,
\end{align}where $\mv{H}_{k,n}=\mv{h}_{k,n}\mv{h}_{k,n}^H$. Moreover, since the non-convex constraint (\ref{eqn:compression constraint 6 eqv 1}) has the same form as constraint (\ref{eqn:constraint 7 eqv 1}) in Section \ref{sec:Proposed Solution}, we can use the convex constraint given in (\ref{eqn:P1 constraint 1}) to approximate it, where $r_k^{\rm DS}$ is substituted by $r_k^{\rm COM}$. As a consequence, the non-convex constraint (\ref{eqn:compression constraint 6}) is approximated by the convex constraints (\ref{eqn:P1 constraint 1}) and (\ref{eqn:P3 constraint 2}).

Next, we deal with the non-convex constraint (\ref{eqn:compression constraint 1}). Since $\log_2|\mv{X}_n|$ is a concave function over $\mv{X}_n\succeq \mv{0}$, its first-order approximation function at any point $\tilde{\mv{X}}_n\succeq \mv{0}$ is an upper bound for it, i.e.,
\begin{align}\label{eqn:compression first-order approximation}
\log_2|\mv{X}_n|\leq \log_2|\tilde{\mv{X}}_n|+\frac{1}{\ln 2}{\rm tr}(\tilde{\mv{X}}_n^{-1}(\mv{X}_n-\tilde{\mv{X}}_n)),
\end{align}where the equality holds if and only if $\mv{X}_n=\tilde{\mv{X}}_n$. By setting $\mv{X}_n=\sum_{k=1}^K\mv{w}_{k,n}\mv{w}_{k,n}^H+\mv{Q}_n$, at any point $\tilde{\mv{X}}_n=\sum_{k=1}^K\tilde{\mv{w}}_{k,n}\tilde{\mv{w}}_{k,n}^H+\tilde{\mv{Q}}_n$, we have
\begin{align}
T_n = & \log_2\frac{\left|\sum\limits_{k=1}^K\mv{w}_{k,n}\mv{w}_{k,n}^H+\mv{Q}_n\right|}{|\mv{Q}_n|} \nonumber \\  \leq & \log_2|\tilde{\mv{X}}_n|+\frac{{\rm tr}\left(\tilde{\mv{X}}_n^{-1}\left(\sum\limits_{k=1}^K\mv{w}_{k,n}\mv{w}_{k,n}^H+\mv{Q}_n-\tilde{\mv{X}}_n\right)\right)}{\ln 2}-\log_2|\mv{Q}_n|  \nonumber \\ = & \log_2|\tilde{\mv{X}}_n|+\frac{\left(\sum\limits_{k=1}^K\mv{w}_{k,n}^H\tilde{\mv{X}}_n^{-1}\mv{w}_{k,n}+{\rm tr}(\tilde{\mv{X}}_n^{-1}\mv{Q}_n-\mv{I})\right)}{\ln 2} -\log_2|\mv{Q}_n|, ~~~ \forall n.
\end{align}As a result, in this paper we approximate the non-convex constraint (\ref{eqn:compression constraint 1}) by the following convex one:
\begin{align}\label{eqn:compression constraint 1 convex}
\log_2|\tilde{\mv{X}}_n|+\frac{1}{\ln 2}\left(\sum\limits_{k=1}^K\mv{w}_{k,n}^H\tilde{\mv{X}}_n^{-1}\mv{w}_{k,n}+{\rm tr}(\tilde{\mv{X}}_n^{-1}\mv{Q}_n-\mv{I})\right) -\log_2|\mv{Q}_n|\leq \sum\limits_{l\in \mathcal{I}(\mathcal{N}_n)}d^n_l, ~~~ \forall n.
\end{align}

To summarize, given $\tilde{\beta}_k$'s, $\tilde{\mv{w}}_{k,n}$'s, and $\tilde{\mv{Q}}_n$'s, problem (\ref{eqn:WSR maximization compression}) is approximated by the following convex problem.
\begin{subequations}\label{eqn:P3}\begin{align}\mathop{\mathrm{maximize}}_{\{\mv{w}_{k,n},r_k^{\rm COM},\eta_k,d^n_l,\mv{Q}_n\}} & ~ \sum\limits_{k=1}^K\mu_kr_k^{\rm COM} \label{eqn:P3 sub}  \\
\mathrm {subject ~ to} ~~~~~~ & ~ (\ref{eqn:compression constraint 5}),  (\ref{eqn:P1 constraint 1}),  (\ref{eqn:P3 constraint 2}),  (\ref{eqn:compression constraint 1 convex}), (\ref{eqn:compression constraint 2}) - (\ref{eqn:compression constraint 4}).
\end{align}\end{subequations}Since problem (\ref{eqn:P3}) is a convex problem, it can be globally solved by CVX. The successive convex approximation method based algorithm to problem (\ref{eqn:WSR maximization compression}) is summarized in Algorithm \ref{table4}, which iteratively updates $\tilde{\beta}_k$'s, $\tilde{\mv{w}}_{k,n}$'s, and $\tilde{\mv{Q}}_n$'s based on the solution to problem (\ref{eqn:P3}) as shown in Step 2). Similar to Section \ref{sec:Proposed Solution}, the convergence behaviour of Algorithm \ref{table4} is guaranteed in the following proposition.

\begin{algorithm}[t]
{\bf Initialization}: Set the initial values for $\tilde{\beta}_k$'s, $\tilde{\mv{w}}_{k,n}$'s, and $\tilde{\mv{Q}}_n$'s, and set $t=1$; \\
{\bf Repeat}:
\begin{enumerate}
\item Find the optimal solution to problem (\ref{eqn:P3}) using CVX as $\{\mv{w}_{k,n}^{(t)},(r_k^{\rm COM})^{(t)},\eta_k^{(t)},(d^n_l)^{(t)},\mv{Q}_n^{(t)}\}$;
\item Update $\tilde{\beta}_k=\sqrt{(2^{(r_k^{\rm COM})^{(t)}/B}-1)/\eta_k^{(t)}}$, $\tilde{\mv{X}}_n=\sum_{k=1}^K\mv{w}_{k,n}^{(t)}(\mv{w}_{k,n}^{(t)})^H+\mv{Q}_n^{(t)}$,  $\forall k,n$;
\item $t=t+1$.
\end{enumerate}
{\bf Until} convergence
\caption{Proposed Algorithm for Solving Problem (\ref{eqn:WSR maximization compression})}
\label{table4}
\end{algorithm}

\begin{proposition}\label{proposition3}
Monotonic convergence of Algorithm \ref{table4} is guaranteed, i.e., $\sum_{k=1}^K\mu_k(r_k^{\rm COM})^{(t)}\geq \sum_{k=1}^K \mu_k(r_k^{\rm COM})^{(t-1)}$. Moreover, the converged solution satisfies all the constraints as well as the KKT conditions of problem (\ref{eqn:WSR maximization compression}).
\end{proposition}

\begin{remark}\label{remark4}
It is worth noting that a similar problem to problem (\ref{eqn:WSR maximization compression}) is studied in \cite{Simeone13}, where the RRHs are assumed to be directly connected to the CP via fronthaul links without routers, and the users are assumed to be equipped with multiple antennas. The successive convex optimization technique is also used to jointly optimize the transmit covariance for each user and compression noise covariance for each RRH so as to maximize the weighted sum-rate of all the users subject to the fronthaul link capacity constraints. Note that in this paper, each user $k$ is assigned with one data stream $s_k$ since it is equipped with one antenna, and the transmit covariance for each user is thus of rank one. As a result, if we optimize the transmit covariance as in \cite{Simeone13} instead of the beamforming vectors, it is necessary to add the rank-one constraints for the transmit covariance matrices, which are non-convex. On the contrary, in this paper we directly optimize the beamforming vector for each user as shown in Algorithm \ref{table4}. The obtained solution is shown to satisfy the KKT conditions of problem (\ref{eqn:WSR maximization compression}).
\end{remark}

\section{Numerical Results}\label{sec:Numerical Results}
\begin{table}
\centering
\begin{center}
\caption{System Parameters of the Numerical Example} \label{table5}
{\small
\begin{tabular}{|c|c|}
\hline Channel Bandwidth & $10$ MHz \\
\hline
Cluster Radius & $1$ km \\
\hline
Number of RRHs & $5$ \\
\hline
Number of Antennas per RRH & $2$ \\
\hline
Number of Users & $10$ \\
\hline
RRH Transmit Power Constraint & $43$ dBm \\
\hline
Antenna Gain & $15$ dBi \\
\hline
Path Loss Model & $128.1+37.6\log_{10}(D)$ dB \\
\hline
Log-Normal Shadowing & $8$ dB \\
\hline
Rayleigh Small Scale Fading & $0$ dB \\
\hline
AWGN Power Spectrum Density & $-169$ dBm/Hz \\
\hline
\end{tabular}
}
\end{center}
\end{table}

In this section, we evaluate the performance of the proposed network coding based data-sharing strategy and routing-based compression-based strategy in the downlink multi-hop C-RAN. In this numerical example, there are $N=5$ RRHs, each equipped with $M=2$ antennas, and $K=10$ users randomly distributed in a circle area of radius $1000$m. The bandwidth of the wireless link is $B=10$MHz. The channel vectors are generated from independent Rayleigh fading, while the path loss model of the wireless channel is given as $128.1+37.6\log_{10}(D)$ in dB, where $D$ (in kilometer) denotes the distance between the user and the RRH. The transmit power constraint for each RRH is $P_n=43$dBm, $\forall n$. The power spectral density of the AWGN at each user receiver is assumed to be $-169$dBm/Hz, and the noise figure due to the receiver processing is $7$dB. The above simulation parameters are summarized in Table \ref{table5}. Moreover, the fronthaul network topology together with the capacities of the fronthaul links (denoted by $2C$ or $C/2$) are shown in Fig. \ref{fig2}. Last, for convenience, the rate weights are assumed to be one for all the users in both problems (\ref{eqn:WSR maximization}) and (\ref{eqn:WSR maximization compression}), i.e., sum-rate maximization is considered for the data-sharing strategy and compression-based strategy.
\begin{figure}
\begin{center}
\scalebox{0.7}{\includegraphics*{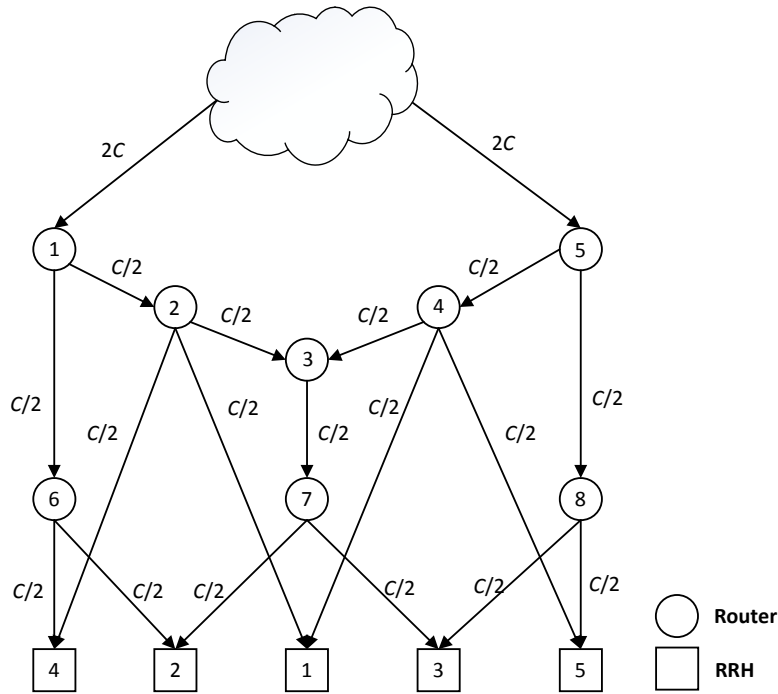}}
\end{center}
\caption{The multi-hop fronthaul topology for C-RAN.}\label{fig2}
\end{figure}

\subsection{Effectiveness of the Proposed Data-Sharing Strategy}
First, we verify the effectiveness of our algorithm proposed in Section \ref{sec:Proposed Solution} to the weighted sum-rate maximization problem (\ref{eqn:WSR maximization}) under the data-sharing strategy. Fig. \ref{fig4} shows the convergence behaviour of the proposed iterative algorithms to problems (\ref{eqn:WSR maximization 1}) and (\ref{eqn:WSR maximization fixed association}), i.e., Algorithms \ref{table1} and \ref{table2}, when $C=200$Mbps and $C=400$Mbps in Fig. \ref{fig2}. Monotonic convergence is observed for both Algorithms \ref{table1} and \ref{table2} with different values of $C$, which verifies Propositions \ref{proposition1} and \ref{proposition2}. Moreover, it is observed that both algorithms converge within $10$ iterations. Last, for both values of $C$, the converged sum-rate of Algorithm \ref{table2} is very close to that of Algorithm \ref{table1}, which verifies that the continuous function $g_{\Phi}(\mv{w}_{k,n})$ given in (\ref{eqn:approximate function}) is a good approximation to the discrete user-RRH association function $\alpha(\mv{w}_{k,n})$ given in (\ref{eqn:user association}) such that the solution to the relaxed problem (\ref{eqn:WSR maximization 1}) is very close to the original problem (\ref{eqn:WSR maximization}).

\begin{figure}
\begin{center}
\scalebox{0.7}{\includegraphics*{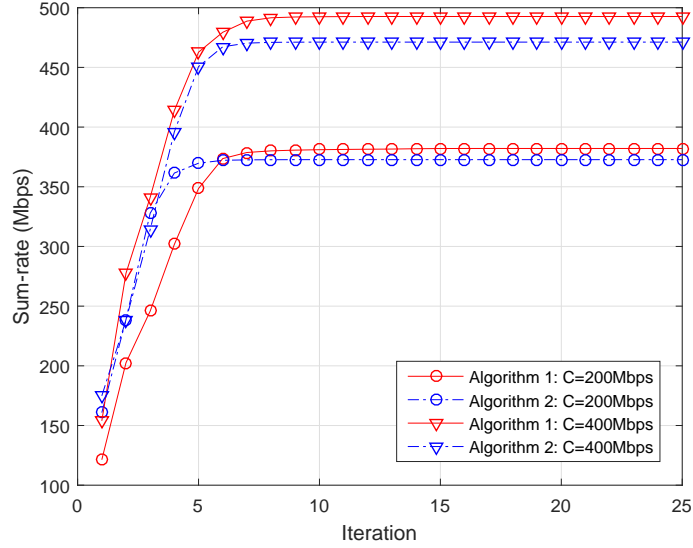}}
\end{center}
\caption{Convergence behaviour of Algorithms $1$ and $2$ in the first and second stages for solving the weighted sum-rate maximization problem (\ref{eqn:WSR maximization}) under the data-sharing strategy.}\label{fig4}
\end{figure}

Next, we verify the effectiveness of our proposed data-sharing strategy. Towards this end, we consider the following three benchmark schemes for performance comparison. For the first benchmark scheme, we consider a strategy where each user is only served by one RRH, as proposed in \cite{Luo14}. Specifically, we first allocate each user to the RRH with the strongest channel power, i.e.,
\begin{align}\label{eqn:single user association}
\alpha_{k,n}=\left\{\begin{array}{ll}1, & {\rm if} ~ n=\arg \max\limits_{1\leq n\leq N} \|\mv{h}_{k,n}\|^2, \\ 0, & {\rm otherwise}, \end{array} \right. \forall k,n.
\end{align}Given the above user-RRH association solution, the CP \mbox{unicasts} each user's data to its associated RRH via routing over the fronthaul network. Note that in a unicast network, the network coding constraints given in (\ref{eqn:constraint 1}) -- (\ref{eqn:constraint 5}) reduce to the unicasting constraints. As a result, the sum-rate of all the users achieved by this scheme can be obtained by solving problem (\ref{eqn:WSR maximization fixed association}) with the user-RRH association solution given in (\ref{eqn:single user association}).

%For Benchmark Scheme 2, we consider the reweighted $\ell_1$-norm based algorithm proposed in \cite{Yu14}. Although the routing constraints (\ref{eqn:constraint 2})--(\ref{eqn:constraint 5}) are not considered in \cite{Yu14}, their proposed algorithm can be modified to solve problem (\ref{eqn:WSR maximization}) since the routing constraints are linear.

For the second benchmark scheme, we allow each user to be served by multiple RRHs. Specifically, we let each user be served by the $3$ RRHs with the first three strongest channel power. Given the above user-RRH association solution, the sum-rate of all the users achieved by this scheme can be obtained by solving problem (\ref{eqn:WSR maximization fixed association}) using Algorithm \ref{table2}.

For the third benchmark scheme, we still let each user be served by the $3$ RRHs with the first three strongest channel power. However, instead of encoding the received information, in this scheme we assume that each router simply replicates and forwards its received information to the other routers in the multi-hop fronthaul network. Note that with the above replicate-and-forward scheme, the routing constraints (\ref{eqn:constraint 1 fixed association}), (\ref{eqn:constraint 2}) -- (\ref{eqn:constraint 5}) in problem (\ref{eqn:WSR maximization fixed association}) need to be modified. Specifically, the multicast of each user's message is built by Steiner trees. Define $\mathcal{T}_k$ as the set of all the Steiner trees for multicasting user $k$'s message, which is determined by the user-RRH association, and $\mathcal{L}_t$ as the set of all the fronthaul links in a Steiner tree $t$. According to \cite{Agarwal04}, the routing constraints for the replicate-and-forward scheme can be formulated as
\begin{align}
& r_k^{\rm DS}\leq \sum\limits_{t\in \mathcal{T}_k}\tau_{t,k}, ~~~ \forall k, \label{eqn:constraint sub 1} \\
& \sum\limits_{k\in \mathcal{K}, t\in \mathcal{T}_k, l\in \mathcal{L}_t} \tau_{t,k}\leq C_l, ~~~ \forall l, \label{eqn:constraint sub 2} \\
& \tau_{t,k}\geq 0, ~~~ \forall t,k, \label{eqn:constraint sub 3}
\end{align}where $\tau_{t,k}$ denotes the rate for multicasting user $k$'s message via Steiner tree $t$. Via replacing the linear constraints (\ref{eqn:constraint 1 fixed association}), (\ref{eqn:constraint 2}) -- (\ref{eqn:constraint 5}) by the linear constraints (\ref{eqn:constraint sub 1}) -- (\ref{eqn:constraint sub 3}) in problem (\ref{eqn:WSR maximization fixed association}), we are able to obtain the sum-rate achieved by the replicate-and-forward based data-sharing strategy.

Fig. \ref{fig3} shows the users' sum-rate achieved by different schemes under the data-sharing strategy versus different values of $C$. It is observed that our proposed data-sharing strategy achieves much higher throughput than its counterpart without cooperation between RRHs, especially when the value of $C$ is large. This is because our proposed scheme provides a joint beamforming design gain. It is also observed that the proposed network coding based scheme provides up to $30\%$ throughput gain as compared to the scheme when each user is served by three RRHs with strongest channel power. This shows that the user-RRH association plays a significant role on the throughput performance and thus should be carefully optimized. Last, it is observed that when each user is served by three RRHs with strongest channel power, the sum-rate achieved by the replicate-and-forward based data-sharing \mbox{strategy} is very close to that achieved by its counterpart based on network coding. This implies that for the information multicast over the fronthaul network, the gain of the network coding technique over the optimized replicate-and-forward scheme is not significant. (Note that similar observations are also found in the literature, e.g. \cite{Wu04}.) However, as shown in Section \ref{sec:Network coding over fronthaul network}, the Steiner tree packing problem arising from the replicate-and-forward scheme is NP-hard, thus from the algorithm design point of view, the network coding technique is preferred.

\begin{figure}
\begin{center}
\scalebox{0.6}{\includegraphics*{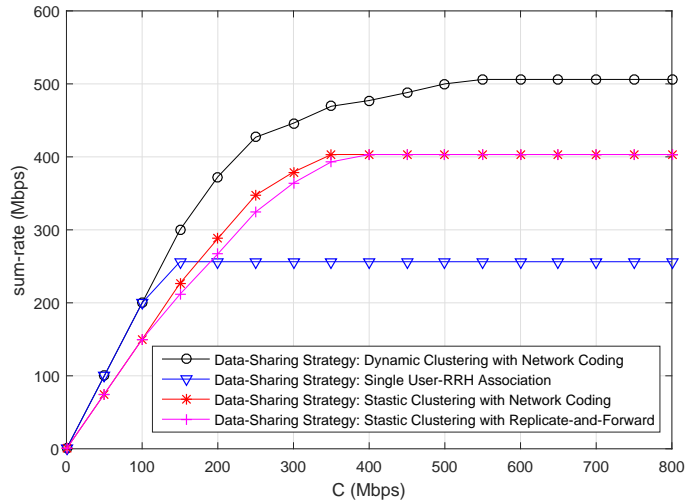}}
\end{center}
\caption{Throughput versus fronthaul link capacity of the data-sharing strategy.}\label{fig3}
\end{figure}

\subsection{Effectiveness of the Proposed Compression-based Strategy}
\begin{figure}
\begin{center}
\scalebox{0.7}{\includegraphics*{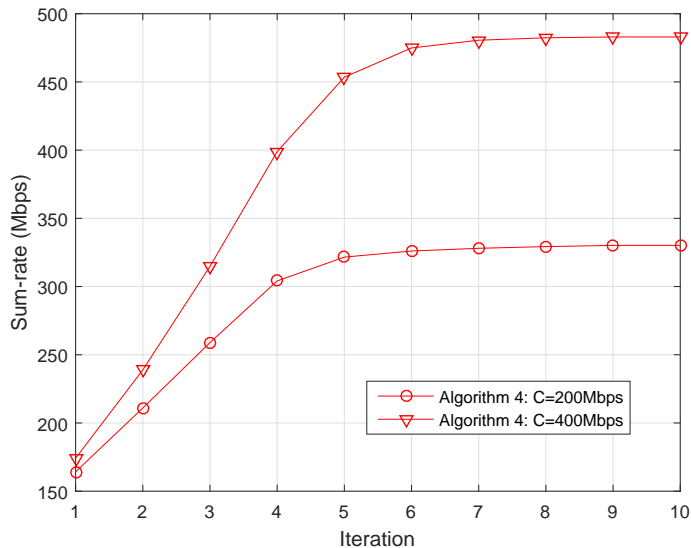}}
\end{center}
\caption{Convergence behaviour of Algorithm \ref{table4} for solving the weighted sum-rate maximization problem (\ref{eqn:WSR maximization compression}) under the compression-based strategy.}\label{fig5}
\end{figure}

In this subsection, we evaluate the performance of the proposed compression-based strategy in the downlink multi-hop C-RAN. Fig. \ref{fig5} shows the convergence behaviour of Algorithm \ref{table4} for problem (\ref{eqn:WSR maximization compression}) when $C=200$Mps and $400$Mbos. Similar to Fig. \ref{fig4}, monotonic convergence is observed for Algorithm \ref{table4}, which verifies Proposition \ref{proposition3}. Moreover, it is observed that Algorithm \ref{table4} converges in less than $10$ iterations for both values of $C$.

\subsection{Comparison between Data-Sharing Strategy and Compression-based Strategy}
\begin{figure}
\begin{center}
\scalebox{0.7}{\includegraphics*{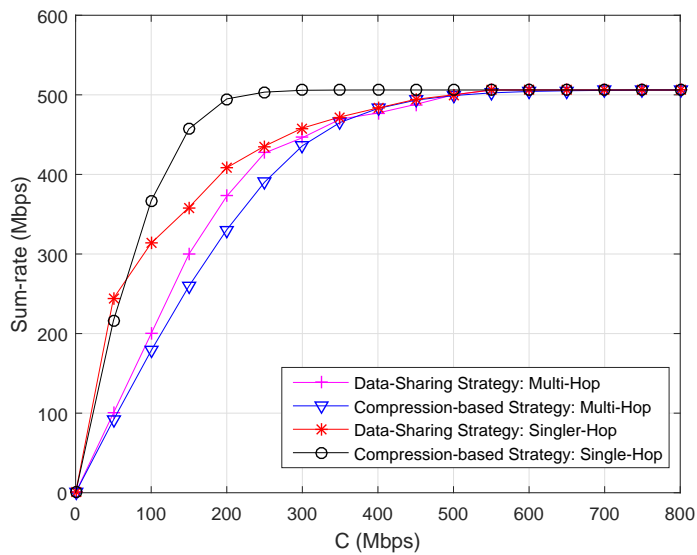}}
\end{center}
\caption{Performance comparison between the data-sharing strategy and compression-based strategy for the multi-hop topology of Fig. \ref{fig2} versus the single-hop topology of Fig. \ref{fig7}.}\label{fig6}
\end{figure}

It is worth noting that the data-sharing strategy and compression-based strategy are two fundamentally different approaches to utilize the fronthaul network in the downlink C-RAN. Under the former strategy, user messages are multicast to the RRHs, while under the latter strategy, each compressed signal is unicast to the corresponding RRH. In this subsection, we aim to answer the following question by simulation results: in the downlink multi-hop C-RAN, which strategy is more efficient for the utilization of the limited capacity in the fronthaul network? Fig. \ref{fig6} provides a performance comparison between the data-sharing strategy and compression-based strategy in terms of the sum-rate of all the users versus the fronthaul link capacity. For the purpose of illustration, we also provide the throughput performance of the data-sharing strategy and compression-based strategy in the case when each RRH is directly connected to the CP via a fronthaul link with capacity $C$, as shown in Fig. \ref{fig7}. Note that in both the setups in Figs. \ref{fig2} and \ref{fig7}, the capacity of the information flow to each RRH is $C$, while the difference is that the routing strategy also influences the throughput performance in the first setup.

It is observed from Fig. \ref{fig6} that in the multi-hop C-RAN, the sum-rate achieved by the data-sharing strategy is higher than that achieved by the compression-based strategy almost for all the values of $C$. Note that this is in sharp contrast to the previous results in \cite{Yu15,Yu16}, which shows that if the routing strategy over the fronthaul network is not considered, in general the compression-based strategy outperforms the data-sharing strategy in terms of both spectral and energy efficiency. Specifically, in this numerical example, it is observed that in the single-hop C-RAN, the compression-based strategy can provide up to $25\%$ performance gain over the data-sharing strategy. By comparing the cases of multi-hop and single-hop C-RAN, it is concluded that although sending the compressed signals is a better option than sending the user messages if the routing strategy is not considered, the data-sharing strategy can utilize the information multicast technique over the fronthaul network, which is more efficient than information unicast of the compression-based strategy, to make up the above disadvantage.

\begin{figure}
\begin{center}
\scalebox{0.7}{\includegraphics*{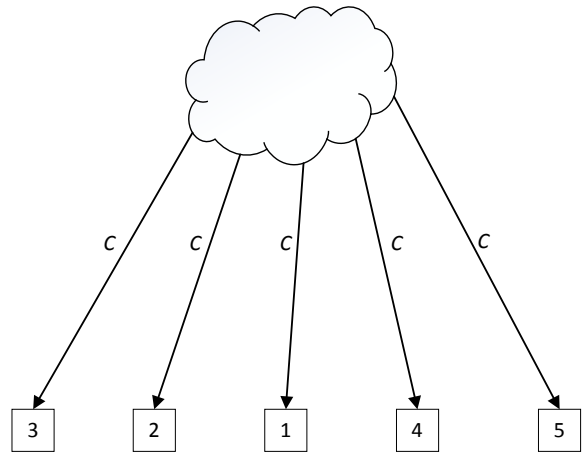}}
\end{center}
\caption{The single-hop C-RAN.}\label{fig7}
\end{figure}

\section{Conclusion}\label{sec:Conclusion}
This paper investigates two fundamentally different techniques for the downlink multi-hop C-RAN, namely data-sharing strategy and compression-based strategy. Different from prior works, apart from the resources in the wireless link, the routing strategy over the multi-hop fronthaul network is considered as well for maximizing the achievable throughput of the downlink C-RAN under both strategies. Specifically, under the data-sharing strategy, the network coding technique is utilized to multicast each user's messages to all the RRHs serving this user, while under the compression-based strategy, a simple routing technique is used to unicast each RRH's compressed signal to the destination. Efficient algorithms with monotonic convergence are proposed under the above cross-layer optimization framework for each strategy, and the obtained solutions are proved to satisfy the KKT conditions of the problems of interests.

Prior works show that if the routing strategy is not considered, the compression-based strategy generally outperforms the data-sharing strategy in terms of spectral efficiency. The main contribution of this paper is that if the routing strategy is jointly optimized with the transmission strategy, the data-sharing strategy can achieve better system throughout than the compression-based strategy in the downlink C-RAN, since information multicast is more efficient than information unicast over the multi-hop fronthaul network. This implies that the data-sharing strategy is also a promising candidate for the downlink communication of the emerging C-RAN.

\begin{appendix}
\subsection{Proof of Proposition \ref{proposition1}}\label{appendix1}
First, it can be shown that in the $t$th iteration of Algorithm \ref{table1}, the solution obtained in the $(t-1)$th iteration is also feasible to problem (\ref{eqn:P1}) given $\tilde{\mv{w}}_{k,n}=\mv{w}_{k,n}^{(t-1)}$, $\tilde{r}_k^{\rm DS}=(r_k^{\rm DS})^{(t-1)}$, and $\tilde{\beta}_k=\sqrt{(2^{r_k^{(t-1)}}-1)/\eta_k^{(t-1)}}$, $\forall k,n$. In other words, $\sum_{k=1}^K\mu_k(r_k^{\rm DS})^{(t-1)}$ is achievable to problem (\ref{eqn:P1}) in the $t$th iteration. As a result, the optimal weighted sum-rate to problem (\ref{eqn:P1}) in the $t$th iteration, i.e., $\sum_{k=1}^K \mu_k(r_k^{\rm DS})^{(t)}$, is no smaller than the optimal weighted sum-rate achieved in the $(t-1)$th iteration, i.e., $\sum_{k=1}^K\mu_k(r_k^{\rm DS})^{(t-1)}$. Monotonic convergence of Algorithm \ref{table1} is thus proved.

Next, since in Algorithm \ref{table1} we use upper-bound to approximate the non-convex functions in problem (\ref{eqn:WSR maximization 1}), as shown in (\ref{eqn:app}), (\ref{eqn:first-order approximation}), and (\ref{eqn:first-order approximation 2}), any feasible solution to problem (\ref{eqn:P1}) satisfies all the constraints of problem (\ref{eqn:WSR maximization 1}). As a result, the solution from Algorithm \ref{table1} must be feasible to problem (\ref{eqn:WSR maximization 1}).

Last, according to \cite[Theorem 1]{Marks78}, if in an optimization problem, each non-convex constraint $f(\mv{x})\leq 0$ is iteratively approximated by a convex constraint $f_{\rm opp}(\mv{x},\tilde{\mv{x}})\leq 0$, where $\tilde{\mv{x}}$ is the optimal solution to the approximated problem in the previous iteration, and $f_{\rm opp}(\mv{x},\tilde{\mv{x}})$ is a convex function satisfying
\begin{align}
& f_{\rm opp}(\mv{x},\tilde{\mv{x}})\geq f(\mv{x}), \label{eqn:condition 1} \\
& f_{\rm opp}(\tilde{\mv{x}},\tilde{\mv{x}})= f(\tilde{\mv{x}}), \label{eqn:condition 2} \\
& \triangledown  f_{\rm opp}(\mv{x},\tilde{\mv{x}})_{|\mv{x}=\tilde{\mv{x}}}=\triangledown f(\mv{x})_{|\mv{x}=\tilde{\mv{x}}}, \label{eqn:condition 3}
\end{align}then the successive convex approximation algorithm can always yield a solution satisfying the KKT conditions of the problem. In the following, we show that constraint (\ref{eqn:P1 constraint 1}) is an approximation to constraint (\ref{eqn:constraint 7 eqv 1}) satisfying the above conditions. First, the inequality (\ref{eqn:app}) implies that $f_{\tilde{\beta}_k}(r_k^{\rm DS},\eta_k)$ is an upper bound to $\sqrt{(2^{r_k^{\rm DS}/B}-1)\eta_k}$, where the equality holds if and only if $\tilde{\beta}_k=\sqrt{(2^{r_k^{\rm DS}/B}-1)/\eta_k}$. Moreover, in Algorithm \ref{table1}, $\tilde{\beta}_k$ is set as $\sqrt{(2^{(r_k^{\rm DS})^{(t)}/B}-1)/\eta_k^{(t)}}$ in each iteration. As a result, the conditions (\ref{eqn:condition 1}) and (\ref{eqn:condition 2}) are satisfied. Next, it can be shown that
\begin{align}
\frac{\partial f_{\tilde{\beta}_k}(r_k^{\rm DS},\eta_k)}{\partial \eta_k} =\frac{\tilde{\beta_k}}{2}=\frac{\sqrt{(2^{r_k^{\rm DS}/B}-1)\eta_k}}{2} =\frac{\partial \sqrt{(2^{r_k^{\rm DS}/B}-1)\eta_k}}{\partial \eta_k}.
\end{align}Similarly, it can be shown that $\partial f_{\tilde{\beta}_k}(r_k^{\rm DS},\eta_k)/\partial r_k^{\rm DS}=\partial \sqrt{(2^{r_k^{\rm DS}/B}-1)\eta_k}/\partial r_k^{\rm DS}$. As a result, constraint (\ref{eqn:P1 constraint 1}) is an approximation to constraint (\ref{eqn:constraint 7 eqv 1}) satisfying the constraints given in (\ref{eqn:condition 1}) -- (\ref{eqn:condition 3}). Moreover, it can be shown that constraint (\ref{eqn:constraint 1 new 1}) is an approximation to constraint (\ref{eqn:constraint 1 log}) satisfying the conditions given in (\ref{eqn:condition 1}) -- (\ref{eqn:condition 3}). As a result, the solution obtained by the successive convex approximation based Algorithm \ref{table1} must satisfy the KKT conditions of problem (\ref{eqn:WSR maximization 1}).

\end{appendix}


\begin{thebibliography}{1}
\bibitem{Simeone16} O. Simeone, A. Maeder, M. Peng, O. Sahin, and W. Yu, ``Cloud radio access network: virtualizing wireless access for dense heterogeneous systems,'' to appear in {\it J. Commun. and Networks}. [Online]. Available: http://arxiv.org/abs/1512.07743

\bibitem{Gesbert11} R. Zakhour and D. Gesbert, ``Optimized data sharing in multicell MIMO with finite backhaul capacity,'' {\it IEEE Trans. Signal Process.}, vol. 59, no. 12, pp. 6102-6111, Dec. 2011.

\bibitem{Yu14} B. Dai and W. Yu, ``Sparse beamforming and user-centric clustering for downlink cloud radio access network,'' {\it IEEE Access}, vol. 2, pp. 1326-1339, 2014.

\bibitem{Zhang16} L. Liu and R. Zhang, ``Downlink SINR balancing in C-RAN under limited fronthaul capacity,'' in {\it Proc. IEEE International Conference on Acoustics, Speech, and Signal Processing (ICASSP)}, Shanghai, China, Mar. 2016.

\bibitem{Simeone13} S. H. Park, O. Simeone, O. Sahin and S. Shamai, ``Joint precoding and multivariate backhaul compression for the downlink of cloud radio access networks,'' {\it IEEE Trans. Signal Process.}, vol. 61, no. 22, pp. 5646-5658, Nov. 2013.


\bibitem{Yeung00} R. Ahlswede, N. Cai, S. R. Li, and R. W. Yeung, ``Network information flow,'' {\it IEEE Trans. Inf. Theory}, vol. 46, no. 4, pp. 1204-1216, Jul. 2000.

\bibitem{Yu15} P. Patil, B. Dai, and W. Yu, ``Performance comparison of data-sharing and compression strategies for cloud radio access networks,'' in {\it Proc. European Signal Processing Conference (EUSIPCO)}, Sept. 2015.

\bibitem{Yu16} B. Dai and W. Yu, ``Energy efficiency of downlink transmission strategies for cloud radio access networks,'' to appear in {\it IEEE J. Sel. Areas Commun.}. [Online]. Available: http://arxiv.org/abs/1601.01070

\bibitem{Luo14} W.-C. Liao, M. Hong, H. Farmanbar, X. Li, Z.-Q. Luo, and H. Zhang, ``Min flow rate maximization for software defined radio access networks,'' {\it IEEE J. Sel. Areas Commun.}, vol. 32, no. 6, pp. 1282-1294, Sept. 2014.

\bibitem{Simeone15} S. H. Park, O. Simeone, O. Sahin, and S. Shamai (Shitz),``Multihop Backhaul Compression for the Uplink of Cloud Radio Access Networks,'' to appear in {\it IEEE Trans. Veh. Techn.}. [online]. Available: http://arxiv.org/abs/1312.7135

\bibitem{Li05} Z. Li, B. Li, D. Jiang, and L. C. Lau, ``On achieving optimal throughput with network coding,'' in {\it Proc. IEEE INFOCOM}, Miami, FL, Mar. 2005, pp. 2184-2194.

\bibitem{Li06} J. Yuan, Z. Li, W. Yu, and B. Li, ``A cross-layer optimization framework for multihop multicast in wireless mesh networks,'' {\it IEEE J. Sel. Areas Commun.}, vol. 24, no. 11, pp. 2092-2103, Nov. 2006.

%\bibitem{Yeung03} S. R. Li, R. W. Yeung, and N. Cai, ``Linear network coding,'' {\it IEEE Trans. Inf. Theory}, vol. 49, no. 2, pp. 371-381, Feb. 2003.

\bibitem{Jaggi05} S. Jaggi, P. Sanders, P. A. Chou, M. Effros, S. Egner, K. Jain, and L. Tolhuizen, ``Polynomial time algorithms for multicast network code construction,'' {\it IEEE Trans. Inf. Theory}, vol. 51, no. 6, pp. 1973-1982, June 2005.

\bibitem{Wu03} P. A. Chou, Y. Wu, and K. Jain, ``Practical network coding,'' in {\it Proc. Annu. Allerton Conf. Communication, Control, and Computing}, Monticello, IL, Oct. 2003.




%\bibitem{Luo11} Q. Shi, M. Razaviyayn, Z.-Q. Luo, and C. He, ``An iteratively weighted MMSE approach to distributed sum-utility maximization for a MIMO interfering broadcast channel,'' {\it IEEE Trans. Singla Process.}, vol. 59, no. 9, pp. 4331-4340, Sep. 2011.
%
%\bibitem{Liu12} L. Liu, R. Zhang, and K. C. Chua, ``Achieving global optimality for weighted sum-rate maximization in the K-User Gaussian interference channel with multiple antennas,'' {\it IEEE Trans. Wireless Commun.}, vol. 11, no. 5, pp. 1933-1945, May 2012.

\bibitem{Luo13} M. Hong, R.-Y. Sun, H. Baligh, and Z.-Q. Luo, ``Joint base station clustering and beamformer design for partial coordinated transmission in heterogenous networks,'' {\it IEEE J. Sel. Areas Commun.}, vol. 31, no. 2, pp. 226-240, Feb. 2013.

\bibitem{Boyd11} M. Grant and S. Boyd, CVX: {\it Matlab software for disciplined convex programming, version 1.21}, http://cvxr.com/cvx/ Apr. 2011.

\bibitem{Agarwal04} A. Agarwal and M. Charikar, ``On the advantage of network coding for improving network throughput,'' in {\it Proc. IEEE 2004 IEEE Information Theory Workshop}, San Antonio, TX,pp. 247-249,  Oct. 2004.

\bibitem{Wu04} Y. Wu, P. A. Chou, and K. Jain, ``A comparison of network coding and tree packing,'' in {\it Proc. IEEE Int. Symp. Inf. Theory}, Chicago, IL, Jun. 2004.

\bibitem{Marks78} B. R. Marks and G. P. Wright, ``A general inner approximation algorithm for nonconvex mathematical programs,'' {\it Operations Research}, vol. 26, no. 4,  pp. 681-683, 1978.

%\bibitem{Boyd04} S. Boyd and L. Vandenberghe, {\it Convex
%Optimization,} Cambridge, U.K., Cambidge Univ. Press, 2004.




\end{thebibliography}
\end{document}